\documentclass[review]{elsarticle}

\usepackage{lineno,hyperref}
\modulolinenumbers[5]

\journal{Systems \& Control Letters}

\usepackage{amsthm}
\usepackage{amsmath} % assumes amsmath package installed
\usepackage{amssymb}  % assumes amsmath package installed
\usepackage{empheq} 
\usepackage{graphicx} % for pdf, bitmapped graphics files
\usepackage{mathptmx} % assumes new font selection scheme installed
\usepackage[]{algorithm2e}

\DeclareMathOperator{\trace}{trace}
\renewcommand{\vec}[1]{\mathbf{#1}}
\newtheorem{theorem}{Theorem}
\newtheorem{proposition}{Proposition}

\newtheorem{lemma}{Lemma}

\usepackage[capitalise]{cleveref}

\graphicspath{{images/}}

\begin{document}

\begin{frontmatter}

\title{Optimal Control with Limited Sensing via\\
Empirical Gramians and Piecewise Linear Feedback}

\author{Atiye~Alaeddini}
%\ead{email@uni.edu}
\address{Institute for Disease Modeling, Bellevue, WA, 98005}

\author{Kristi~A.~Morgansen}
%\address{University of Washington, William E. Boeing Department of Aeronautics and Astronautics, Seattle, WA 98195}
%\fntext[myfootnote2]{Senior~Member,~IEEE}
\author{Mehran~Mesbahi}
\address{University of Washington, William E. Boeing Department of Aeronautics and Astronautics, Seattle, WA 98195}
%\fntext[myfootnote3]{Fellow,~IEEE}


\begin{abstract}
This paper is concerned with the design of optimal control for finite-dimensional control-affine nonlinear dynamical systems. We introduce an optimal control problem that specifically optimizes nonlinear observability in addition to ensuring stability of the closed loop system. A recursive algorithm is then proposed to obtain an optimal state feedback controller to maximize the resulting non-quadratic cost functional. The main contribution of the paper is presenting a control synthesis procedure that provides closed loop asymptotic stability,
on one hand, and empirical observability of the system, as a transient performance criteria, on the other.
\end{abstract}

\begin{keyword}
%\texttt{elsarticle.cls}\sep \LaTeX\sep Elsevier \sep template
%\MSC[2010] 00-01\sep  99-00
Nonlinear observability\sep optimal state feedback\sep stability\sep gradient-based algorithm
\end{keyword}

\end{frontmatter}

\linenumbers

\section{Introduction}

A primary consideration in the control synthesis examined in this
paper is system observability. 
This is motivated by the inherent coupling in the
actuation and sensing in nonlinear systems; as such, one might be able to make a system more observable by changing the control inputs \cite{Alaeddini13,hinson2013observability,hinson2013path}. 

Our goal in this work is designing a process for efficiently choosing the inputs of a nonlinear
system to improve the overall sensing performance (measured by the empirical Gramian).
This task is motivated through problems where the reconstruction of the state is challenging
and it becomes imperative to actively acquire, sense, and estimate it; an example
of such a problem is robot localization operating in an unknown environment. 
The problem of active sensing has received significant attention in the robotics community in recent years. 
In this context, a robot can use the path that maximizes observability to more effectively localize itself or to effectively construct a map of the environment \cite{DeVries13, Lorussi01}. 
In \cite{DeVries13}, sampling trajectories for autonomous vehicles are selected to maximize observability of an environmental flow field. In that work, an iterative technique has been used on a finite sampling of parameters instead of examining all possible sampling trajectories. 
Moreover, the resulting optimized path was computed by an exhaustive search. 
In \cite{Lorussi01}, an optimal path is given for a particular case of omnidirectional vehicle moving in a planar environment with only two markers.
Yu {\em et al.}~\cite{Yu11} have developed a path planning algorithm based on dynamic programming for navigation of a micro aerial vehicle (MAV) using bearing-only measurements. 
In their proposed algorithm, at each time step, one of the three possible choices of roll command is selected by solving an optimization problem based on an observability criterion for the chosen path. 
Each of these approaches are applicable for a specific type of nonlinear systems; solving the problem of choosing controls to improve the observability of an arbitrary nonlinear system in general is
a challenging problem.

The fundamental step toward designing an observability-based optimal controller is to formulate the goal of control as a long-term optimization of a scalar cost function. In order to formulate this control design problem using optimal control, one must define an objective that be used to evaluate the long-term performance of running each candidate control policy, $\vec{u}(\vec{x})$ (from each initial condition $\vec{x}_0$), subject to relevant algebraic and dynamic constraints.
Several measures have been used in the literature to evaluate the observability of a system. In the case of the observability-based optimization problem, the observability Gramian has been the most popular notion used to evaluate observability. In order to perform an optimization on the observability Gramian matrix, a scalar function of the matrix must be chosen. The smallest singular value \cite{waldraff1998use,Krener09}, the determinant \cite{alaeddini2016optimal,wouwer2000approach,Serpas13}, the trace \cite{alaeddini2017adaptive,van2000selection}, and the spectral norm \cite{Singh05} of the observability Gramian are candidate measures for observability optimization. 
However, the analytical computation of the nonlinear observability Gramian is expensive and requires that the system dynamics be described by smooth functions. An alternative tool used to measure the level of observability for nonlinear systems is the empirical observability Gramian. To compute the empirical observability Gramian, one only needs to collect empirical data (either from experiments or simulation).

For nonlinear systems, optimal state feedback can be obtained from the solution of the Hamilton-Jacobi-Bellman (HJB) equation. This equation is usually hard to solve analytically. Over the last few years, with increasingly powerful computers and efficient algorithms, the popularity of numerical optimal control has considerably grown. In this direction, different approximation methods have been used to obtain suboptimal feedback control policies for quadratic \cite{Beeler00, Chen03} or non-quadratic \cite{Cimen04, Esfahani11} cost functionals. Most of these solution methods involve solving Riccati equations and a series of linear algebraic or differential equations. In this work, we use a recursive algorithm to find a solution to our optimization problem that maximizes the system observability. Along the way, we approximate the solution to the optimal control with piecewise linear controls.

In comparison with the previous observability-based path planning algorithms reported in the literature, 
the state feedback obtained in this paper can be applied to a larger family of nonlinear systems; the algorithm also ensures the closed loop stability. In particular, utilizing the concept of observability Gramian, and tools from optimal control theory, this paper develops an algorithm for suboptimal piecewise linear state feedback controller with stability guarantees. The objective function examined in this paper is quadratic in state and control, augmented with an observability measure. One of our contribution is presenting an observability index, whose optimization does not have adverse effects on the stability of the system.

\section{Mathematical Preliminaries}
\label{sec:prelim}
Consider a linear time invariant system of the form
\begin{eqnarray}
\begin{aligned}
   \dot{\vec{x}}& = A \vec{x} + B \vec{u}  \\
   \vec{y} &= C \vec{x} \,,  \label{linear}
\end{aligned} 
\end{eqnarray}
for which, the observability Gramian,
\begin{equation}
	W_{o,L} = \int_{0}^{t_f} e^{A^Tt} C^T C e^{At} \mathrm{d}t \,, \label{LinGram}
\end{equation}
can be computed to evaluate the observability of a linear system \cite{Krener09}. 
A linear system is (fully) observable if and only if the corresponding observability Gramian is full rank \cite{Muller72}. In the linear setting, if the observability Gramian is rank deficient, certain states (or directions in the state space) cannot be reconstructed from the sensor data, regardless of the control policy being applied.
%
%\subsection{Empirical Observability Gramian}
%\label{subsec:Empgram}
%
While the observability Gramian works well for determining the observability of linear systems, analytical observability conditions for nonlinear systems quickly become intractable, necessitating 
modifications for computational tractability.
One approach is evaluation of nonlinear observability Gramian through linearization
and the corresponding Jacobian matrices; however, this approach only provides
an approximation of the local observability for a specific trajectory. 
One alternative method to evaluate observability of a nonlinear system is using the relatively new concepts of the observability covariance or the empirical observability Gramian \cite{Krener09}. 
This approach provides a more accurate description of a nonlinear system's observability, while being less computationally expensive than analytical tools such as Lie algebraic based methods. 
To set the stage for the main contribution of the paper, 
let us first provide a brief overview of the empirical observability Gramian and how
it is used to evaluate the observability of nonlinear systems.

Consider the problem of system observability for nonlinear systems in control affine form,
\begin{empheq}[left=\empheqlbrace ]{align}
 &    \dot{\vec{x}}=\vec{f}_0(\vec{x})+\sum_{i=1}^{p} \vec{f}_i(\vec{x})u_i, \ \  \vec{x} \in \mathbb{R}^n, \nonumber \\
 &	 \vec{y}=\vec{h}(\vec{x}),  \ \ \vec{y} \in \mathbb{R}^m. \label{UAffine}
\end{empheq}
For such systems, $\vec{f}_0(\vec{x})$ is referred to as the drift vector field or simply drift, and $\vec{f}_i(\vec{x})$ are termed the control vector fields. 

The empirical observability Gramian for the nonlinear system \eqref{UAffine} is constructed as follows. For $\epsilon > 0$, let $\vec{x}_0^{\pm i} = \vec{x}_0 \pm \epsilon \vec{e}_i$ be the perturbations of the initial condition and $\vec{y}^{\pm i}(t)$ be the corresponding output, where $\vec{e}_i$ is the $i^{\text{th}}$ unit vector in $\mathbb{R}^n$. For the system \eqref{UAffine}, the {\em empirical observability Gramian} $W_o$, is an $n \times n$ matrix, whose $(i,j)$ entry is 
\begin{equation}
	W_{o_{ij}} =\frac{1}{4\epsilon^2} \int_{0}^{t_f} \left(\vec{y}^{+i}(t)-\vec{y}^{-i}(t) \right)^T \left(\vec{y}^{+j}(t)-\vec{y}^{-j}(t) \right) \mathrm{d}t. \label{EmpObsGram}
\end{equation}
It can be shown that if the system dynamics ($\vec{f}_0$ and $\vec{f}_i$ functions in \eqref{UAffine}) are smooth, then the empirical observability Gramian converges to the local observability Gramian as $\epsilon \to 0$. Note that the perturbation, $\epsilon$, should always be chosen such that the system stays in the region of attraction of the equilibrium point of the system. The largest singular value \cite{Singh05}, the smallest eigenvalue \cite{Krener09}, the determinant \cite{DeVries13, Serpas13}, and the trace of the inverse \cite{DeVries13} of the observability Gramian have been used as different measures for the observability.

The empirical observability Gramian \eqref{EmpObsGram} is used to introduce an index for 
measuring the degree of observability of a given nonlinear system. This index is called the \emph{observability index}; here we will use the trace of the empirical observability Gramian for the measure of the observability, i.e., 
\begin{equation}
	\trace(W) = \frac{1}{4\epsilon^2} \int_{0}^{t_f} \sum\limits_{i=1}^n \left\| \vec{h}(\vec{x}^{+i} (t))-\vec{h}(\vec{x}^{-i}(t)) \right\|^2 \mathrm{d}t\,, \label{TraceGram}
\end{equation}
where, $\vec{x}^{\pm i} (t)$ is the trajectory corresponding to the initial condition $\vec{x}_0^{\pm i} = \vec{x}_0 \pm \epsilon \vec{e}_i$.
%The observability Gramian used here is a little bit different from \eqref{EmpObsGram} and comes from the definition empirical observability Gramian presented in \cite{lall2002subspace}. However, these two definitions are very similar, and both have been used in different applications.

\section{Problem Formulation}
\label{sec:prob}

In order to preserve the asymptotic convergence of the closed loop state trajectory, we consider an augmented optimal control problem. In this augmented control problem, the control objective is twofold:
\begin{itemize}
\item Improve the observability of the system in order to avoid unobservable trajectories, and
\item Guarantee the stability of the system and convergence of the resulting state trajectory. 
\end{itemize}
In this direction, we consider the following total cost functional
\begin{equation}
\begin{aligned}
& \underset{\vec{x}, \vec{u}}{\text{min}}
& & J = \int_{0}^{t_f} \left\{ l_1(\vec{x},\vec{u})-l_2(t,\vec{x},\vec{x}^{\pm 1}, \cdots, \vec{x}^{\pm n}) \right\} \mathrm{d}t\,,
\end{aligned} \label{CostToGo}
\end{equation}
\begin{equation}
\begin{aligned}
	&l_1(\vec{x},\vec{u})=\vec{x}^T Q \vec{x}+\vec{u}^TR\vec{u} \\
	&l_2(t,\vec{x},\vec{x}^{\pm 1}, \cdots, \vec{x}^{\pm n}) = e^{-t} \text{sat}_{\zeta} \left( \frac{1} {4\epsilon^2} \sum\limits_{i=1}^n  \left\| \vec{h}(\vec{x}^{+i})-\vec{h}(\vec{x}^{-i}) \right\|^2 \right)\,, \label{ObsIdx}
\end{aligned} 
\end{equation}
where, $\text{sat}_{\zeta}(x)$ is 
\begin{equation} \text{sat}_{\zeta}(x) = \left\{ \begin{aligned}
 &    \zeta \ \  \text{if } x>\zeta \\
 &	  x \ \  \text{if }x \leq \zeta \end{aligned} \right. \label{satFunc}
\end{equation}
and $\zeta > 0$ is chosen such that the optimization problem (\ref{CostToGo}) is feasible and meaningful.

The proposed cost functional consists of two parts: 
\begin{itemize}
\item The first term, $l_1(\vec{x},\vec{u})$, takes into account the control energy and the deviation from the desired steady state trajectory; in our setup we assume that $Q$ and $R$ are both positive definite. 
\item The observability index, $l_2(t,\vec{x},\vec{x}^{\pm 1}, \cdots, \vec{x}^{\pm n})$, which is a transient term, taking into account the local observability of the system.
%, given by 
%\begin{equation}
%	l_2(t,\vec{x},\vec{x}^{\pm 1}, \cdots, \vec{x}^{\pm n}) = e^{-t} \text{sat}_{\zeta} \left( \frac{1} {4\epsilon^2} \sum\limits_{i=1}^n  \left\| \vec{h}(\vec{x}^{+i})-\vec{h}(\vec{x}^{-i}) \right\|^2 \right), \label{ObsIdx}
%\end{equation}
%where, $\text{sat}_{\zeta}(x)$ is 
%\begin{equation} \text{sat}_{\zeta}(x) = \left\{ \begin{aligned}
% &    \zeta \ \  \text{if } x>\zeta \\
% &	  x \ \  \text{if }x \leq \zeta \end{aligned} \right. \label{satFunc}
%\end{equation}
%and $\zeta > 0$ is chosen such that the optimization problem (\ref{CostToGo}) is feasible and meaningful.
\end{itemize}

In the cost function \eqref{CostToGo}, the sum of the control effort (i.e., the integral of the $\vec{u}^T R \vec{u}$ term) and the deviation from the desired trajectory (i.e., the integral of the $\vec{x}^T Q\vec{x}$ term) is minimized, while maximizing an observability index. 
The index $l_1(\vec{x},\vec{u})$ determines the asymptotic behavior of the closed loop system, while $l_2(t,\vec{x},\vec{x}^{\pm 1}, \cdots, \vec{x}^{\pm n})$ specifies the desired transient behavior of the system. This cost function does not directly maximize observability.
Instead the observability term, $l_2(t,\vec{x},\vec{x}^{\pm 1}, \cdots, \vec{x}^{\pm n})$, tunes the cost function so that the obtained optimal control makes the system more observable. We note that maximizing the index of observability and ignoring the remaining terms does not guarantee the stability of system. 
In order to have a meaningful optimization problem, it is desired to have $\vec{x}^T Q \vec{x}-l_2(\cdot) \geq 0$. The saturation limit $\zeta$ in \eqref{satFunc}, ensures that $\vec{x}^T Q \vec{x}-l_2(\cdot)$ remains a positive semi-definite function. However, the saturation limit can be set to a constant number, for example, when we have information on the rate of convergence of the system. Given a nonlinear system $\dot{\vec{x}} = \vec{f}(\vec{x},\vec{u})$, assume that there is a lower-bound, $\beta$, on the rate of decay of the nonlinear system, such that
\begin{equation}
\|\vec{x}(t)\|_Q \geq e^{-\beta t} \|\vec{x}(t_j)\|_Q\,, \ \ \forall t \in [t_j, t_{j+1})\,, \label{RateDecayBound}
\end{equation}
where $\|\cdot\|_Q$ is the norm of the vector induced by the positive definite matrix $Q$. In this case, the saturation parameter, $\zeta$, could be set to
\begin{equation}
\zeta = \left\{
\begin{aligned}
	&\|\vec{x}(t_j)\|_Q^2 & \text{if} && 0<\beta \leq \frac{1}{2},  \\
	&e^{(1-2\beta)t_{j+1}} \|\vec{x}(t_j)\|_Q^2 & \text{if} && \beta>\frac{1}{2}\,.
\end{aligned} \right. \label{saturatePars}
\end{equation}

\begin{lemma}
Consider the nonlinear system $\dot{\vec{x}} = \vec{f}(\vec{x},\vec{u})$. If condition \eqref{RateDecayBound} is satisfied, then the saturation parameter $\zeta$, given by \eqref{saturatePars}, guarantees that $l_0(\vec{x}) = \vec{x}^T Q \vec{x}-l_2(\cdot) \geq 0$.
\end{lemma}
\begin{proof}
From \eqref{saturatePars} we have
\begin{equation}
\zeta \leq e^{(1-2\beta)t} \|\vec{x}(t_j)\|_Q^2, \ \ \forall t \in [t_j, t_{j+1})\,.
\end{equation}
On the other hand, using \eqref{RateDecayBound}, one can conclude that
\begin{equation}
\zeta \leq e^t \|\vec{x}(t)\|_Q^2\,.
\end{equation}
From the definition of the saturation function~\eqref{satFunc}, we have $\text{sat}_\zeta(\cdot) \leq \zeta$. 
Hence, 
\begin{equation}
\begin{aligned}
&l_2(\cdot) = e^{-t} \text{sat}_\zeta(\cdot) \leq \|\vec{x}(t)\|_Q^2, 
\end{aligned}
\end{equation}
and $\vec{x}^T Q \vec{x} - l_2(\cdot) \geq 0$, for all $t \in [t_j, t_{j+1})$.
\end{proof}

The problem considered here is a nonlinear optimal control with a non-quadratic cost function \eqref{CostToGo}.
This problem is generally solved by solving a partial differential equation (Hamilton-Jacobi-Bellman (HJB) equation). 
In general, HJB does not have an analytical solution and potentially difficult to solve computationally.
In order to simplify the derivation of the control input, 
the original problem is divided into a sequence of subproblems of finding an optimal control in the state feedback form, namely $\vec{u}=K \vec{x}$, over shorter time intervals; as such,
the resulting control input can be viewed as a sub-optimal solution to the original problem. 
This approach motivates designing a {\em piece-wise linear state feedback} for the augmented optimal control
problem. Let $0 = t_0 < t_1 < \cdots < t_j < t_{j+1} < \cdots < t_f$. 
Then the optimal linear state feedback control $\vec{u}_j$ for the time interval $t \in [t_j\ \ t_{j+1}), \ \ j=0,\cdots, f-1$ is given by
\begin{equation}
\begin{aligned}
\vec{u}_j^* = & \underset{\vec{x}, \vec{u}}{\text{argmin}}
& & \int_{t_j}^{t_{j+1}} \left\{ l_1(\vec{x},\vec{u})-l_2(t,\vec{x},\vec{x}^{\pm 1}, \cdots, \vec{x}^{\pm n}) \right\} \mathrm{d}t + L(\vec{x}(t_{j+1}),t_{j+1}) \\
& \text{subject to}
& & \vec{u} =K_j \vec{x},\ \ t \in [t_j\ \ t_{j+1}) \,,
\end{aligned} \label{CostTran}
\end{equation}
where, $L(\vec{x}(t_{j+1}),t_{j+1})$ is the cost-to-go at the final time of the interval $[t_j\ \ t_{j+1} )$. Without loss of generality, consider a time interval $[t_j \ \ t_{j+1} )$. By the substitution of $\vec{u}=K_j \vec{x}$ into \eqref{CostTran}, the cost functional can then be approximated as,
\begin{equation}
   J(K_j)  = \int_{t_j}^{t_{j+1}} \Gamma(t,\vec{x},\vec{x}^{\pm 1}, \cdots, \vec{x}^{\pm n},K_j) \mathrm{d}t + L(\vec{x}(t_{j+1}),t_{j+1})\,, \label{Cost}
\end{equation}
where
\begin{equation}
\Gamma(t,\vec{x},\vec{x}^{\pm 1}, \cdots, \vec{x}^{\pm n},K_j) = \vec{x}^T \left( K_j^TR K_j + Q \right) \vec{x}-l_2(t,\vec{x},\vec{x}^{\pm 1}, \cdots, \vec{x}^{\pm n}) \,.
\end{equation}
Now, considering the dynamics of the nonlinear system, the optimization problem can be expressed as
\begin{equation} 
\begin{aligned}
& \underset{K_j}{\text{minimize}}
& & J(K_j) = \int_{t_j}^{t_{j+1}} \Gamma(t,\vec{x},\vec{x}^{\pm 1}, \cdots, \vec{x}^{\pm n},K_j) \mathrm{d}t+ L(\vec{x}(t_{j+1}),t_{j+1}) \\
& \text{subject to}
& & \dot{\vec{x}}=\vec{f}(\vec{x},K_j), \ \ \vec{x}(t_j) = \vec{x}_j,\\
& 
& & \dot{\vec{x}}^{\pm k}=\vec{f}(\vec{x}^{\pm k},K_j), \ \ \vec{x}^{\pm k}(t_j) = \vec{x}_j \pm \epsilon \vec{e}_k\,, \ \ k=0,\cdots,n\,.
\end{aligned} \label{ModCost}
\end{equation}
The optimal gain $K_j^*$ need to be found to minimize this cost-to-go $J$. 

%%%%%%%%%%%%%%%%%%%%%%%%%%%%%%%%%%%%%%%%%%%%%%%%%%%%%%%%%%%%%%%%%%%%%%%%%%%
\section{Optimization Solution}
\label{sec:optimal}
%\subsection{Optimization Algorithm}

In this section, an algorithm is presented to solve \eqref{ModCost} to determine the control policy for the optimal control objective that has an embedded observability index. Here, a recursive gradient-based algorithm is devised to find a solution to this optimization problem. 

Given \eqref{Cost}, we first define a new variable, $x_{n+1}$,
\begin{equation}
    x_{n+1}(t)= \int_{t_j}^{t} \Gamma(\tau,\vec{x}(\tau),\vec{x}^{\pm 1}(\tau), \cdots, \vec{x}^{\pm n}(\tau),K_j) d\tau+ L(\vec{x}(t),t).  \label{NewVar}
\end{equation}
Assuming $\vec{x}(t_j)=\vec{x}_j$ is given, it is clear that $x_{n+1}(t_j) = L(\vec{x}_j,t_j)$ and $x_{n+1}(t_{j+1}) = J(K_j)$. 
Next, define an augmented state vector as 
$$\vec{\bar{x}}(t)= \begin{bmatrix} \vec{x}(t) \\ \vec{x}^{\pm 1}(t) \\ \vdots \\ \vec{x}^{\pm n}(t) \\ x_{n+1}(t) \end{bmatrix}\,.$$ 
Then,
\begin{equation}
\dot{\bar{\vec{x}}} = \begin{bmatrix} \dot{\vec{x}} \\ \dot{\vec{x}}^{\pm 1} \\ \vdots \\ \dot{\vec{x}}^{\pm n} \\ \dot{x}_{n+1} \end{bmatrix} = \mathbb{H}(t,\vec{\bar{x}},K_j)\,; \ \ \vec{\bar{x}}(t_j) = \begin{bmatrix} \vec{x}_j \\ \vec{x}_j\pm \epsilon \vec{e}_1 \\ \vdots \\ \vec{x}_j\pm \epsilon \vec{e}_n \\  L(\vec{x}_j,t_j) \end{bmatrix}, \label{AugStateEq}
\end{equation}
where,
\begin{equation}
	\mathbb{H}(t,\vec{\bar{x}},K_j) = \begin{bmatrix} \vec{f}(\vec{x},K_j) \\ \vec{f}(\vec{x}^{\pm 1},K_j) \\ \vdots \\ \vec{f}(\vec{x}^{\pm n},K_j) \\ \Gamma(t,\vec{x}, \vec{x}^{\pm 1}, \cdots, \vec{x}^{\pm n},K_j)+\dot{L}(\vec{x},t) \end{bmatrix}\,,
\end{equation}
and $\vec{f}(\vec{x},K_j)$ is given in \eqref{UAffine} after substituting $\vec{u}=K_j \vec{x}$. The optimization problem can now be formulated as:
\begin{equation}
\begin{aligned}
& \underset{K_j}{\text{minimize}}
& & J = x_{n+1}(t_{j+1}) \\
& \text{subject to}
& & \dot{\bar{\vec{x}}} = \mathbb{H}(t,\vec{\bar{x}},K_j).
\end{aligned}
\end{equation}

In order to solve this optimization problem, the gradient vector, $\displaystyle \frac{\partial J}{\partial K_j} = \left. \frac{\partial x_{n+1}}{\partial K_j} \right |_{t=t_{j+1}}$, is required. In \cite{Esfahani11}, a method is demonstrated for solving a similar problem setup. By defining $\displaystyle \bar{X}_{K_j} = \frac{\partial \vec{\bar{x}}}{\partial K_j}$, and applying the chain rule, we have
\begin{eqnarray}
&&    \dot{\bar{X}}_{K_j}= \frac{\partial \mathbb{H}}{\partial \vec{\bar{x}}} \frac{\partial \vec{\bar{x}}}{\partial K_j} +\frac{\partial \mathbb{H}}{\partial K_j} =  \frac{\partial \mathbb{H}}{\partial \vec{\bar{x}}} \bar{X}_{K_j} +\frac{\partial \mathbb{H}}{\partial K_j},  \nonumber \\
&&     \bar{X}_{K_j}(t_j) = \vec{0}.\label{XK_0_dot}
\end{eqnarray}
Notice that if the system is multi-input, then $K_j$ is a matrix and the term $\displaystyle \frac{\partial \vec{\bar{x}}}{\partial K_j}$ is the derivative of a vector with respect to a matrix, which results in a higher order tensor. 
In order to address multi-input as well as single-input control problems using more conventional linear algebra, we need to have separate equations for $\displaystyle \bar{X}_{\vec{k}_{ij}} = \frac{\partial \vec{\bar{x}}}{\partial \vec{k}_{ij}}, i=1 \hdots p$, where $\vec{k}_{ij}$ is the $i^{th}$ row of the matrix $K_j$.

Now \eqref{AugStateEq} and \eqref{XK_0_dot} should be solved together for $t_j<t<t_{j+1}$. The last row of $\bar{X}_{K_j}$ at $t=t_{j+1}$ is the gradient vector which can be used for improving the optimal value of $K_j$. Based on these observations, an iterative algorithm for obtaining the optimal feedback gain, $K_j^*$, is presented in \cref{GDecentAlg}. Note that we are using an estimate of the Hessian matrix to make sure that the optimal point is a minima of the objective (and not a saddle point). One can use a more precise method (e.g.\, a finite difference approximation for estimating the Jacobian presented in \cite{spall2000adaptive}). Since we are not directly using the Hessian in the computation of the optimal point, we are using a simple method to estimate the curvature of the function at each iterative step.
\begin{algorithm}[!h]
Choose step size $\mu_0 > 0$\;
Choose $K_j^0$, an initial value for $K_j$ that stabilizes the system\;
cvxCheck : = TRUE\;
 Choose a convergence tolerance $\epsilon_t$\;
 \Repeat{$\|\vec{g}_i\| \leq \epsilon_t$ AND cvxCheck is TRUE}{
  Solve \eqref{AugStateEq} and \eqref{XK_0_dot} together for $t_j<t<t_{j+1}$\;
  $\vec{g}_i$ := the last row of $\bar{X}_{K_j}$ at $t=t_{j+1}$\;
  $K_j^{i+1} = K_j^i - \mu_i \vec{g}_i$\;
  Calculate $H$ such that $H[mn] = \frac{\vec{g}_i[m]-\vec{g}_{i-1}[m]}{K_j^{i+1}[n]-K_j^i[n]}$\;
  \eIf{$H \succeq 0$}{
  cvxCheck := TRUE \;
  }{
  cvxCheck := FALSE \;
  }  
  \eIf{cvxCheck is TRUE}{
  Update step size $\mu_i$\;
  }{
  Do not change step size\;
  }
 }
 \caption{Iterative algorithm to find $K_j^*$} \label{GDecentAlg}
\end{algorithm}

The iteration terminates when the gradient vector becomes small enough or remains almost unchanged. 
The algorithm also checks the convexity of the objective at the solution obtained in order to ensure that it corresponds to a local minimum. The value of $K_j$ is updated based upon the gradient descent rule. The step size $\mu_i$ is a positive number, chosen small enough such that system stability is guaranteed.

%%%%%%%%%%%%%%%%%%%%%%%%%%%%%%%%%%%%%%%%%%%%%%%%%%%%%%%%%%%%%%%%%%%%%%%%%%%
\subsection{Convergence of the Algorithm}
\label{sec:conv}

The gradient method is a popular first order method due to its simplicity. 
The algorithm requires the computation of the gradient, but not the second derivatives. One of the parameters that must be chosen in this method is the step size. 

\begin{proposition}
Any direction that makes an angle of strictly less than $\frac{\pi}{2}$ with $-\nabla J$ is guaranteed to produce a decrease in $J$ provided that the step size is sufficiently small \cite{Nocedal99}.
\end{proposition}

%\begin{proof}
%Suppose that the cost function $J$, given in \eqref{Cost}, is continuously differentiable. Then Taylor's theorem can be used to prove this proposition. The details of this proof can be found in Wright and Nocedal \cite{Nocedal99}.
%\end{proof}

The suggested rule for the step size selection here is 
\begin{equation}
	\mu_i = \frac{\mu_0}{i},
\end{equation}
where, $0 <  \mu_0 < \infty$ is constant, and $i$ is the iteration counter. This step size policy, which is called \emph{square summable but not summable}, satisfies
\begin{equation}
	\begin{aligned}
		\sum_{i=1}^{\infty} \mu_i &= \sum_{i=1}^{\infty} \frac{\mu_0}{i} = \mu_{0} \sum_{i=1}^{\infty} \frac{1}{i} = \infty \\
		\sum_{i=1}^{\infty} \mu_i^2 &= \sum_{i=1}^{\infty} \frac{\mu_0^2}{i^2} = \mu_{0}^2 \sum_{i=1}^{\infty} \frac{1}{i^2} < \infty.
	\end{aligned} \label{SSNS}
\end{equation}

It is worth noting that because of the transient term, $l_2(\cdot)$, the cost functional $J(K)$ is not convex. But since the transient term is bounded, the cost functional is bounded by two convex functions. The lower and upper bounds for the cost functional is given by
\begin{equation}
	\begin{aligned}
		\underline{J} &= \int_{t_j}^{t_{j+1}} \left\{ l_1(\vec{x},\vec{u})-\zeta \right\} \mathrm{d}t + L(\vec{x}(t_{j+1}),t_{j+1})\,, \\
		\bar{J} &= \int_{t_j}^{t_{j+1}} l_1(\vec{x},\vec{u}) \mathrm{d}t + L(\vec{x}(t_{j+1}),t_{j+1})\, .
	\end{aligned} \label{JupJdown}
\end{equation}
Both of the lower and upper bounds ($\underline{J}$ and $\bar{J}$) are convex functions. Therefore, the cost function cannot be concave over the entire search space and there should exist regions over which the cost function is locally convex. We assume here that after a finite number of iterations, the algorithm drives the system to a locally convex point and after that, it remains in a convex region around that point until the algorithm converges to a local minimum.
%\begin{figure}[!h]
%   \begin{center}
%  	 {\includegraphics[width=0.4\textwidth]{UpDown2.pdf}}
%%	 \vspace{-6 mm}
%   \end{center}
%   \caption{Illustrative figure of the upper and lower bounds of the cost function.}
%   \label{UpDown}
%\end{figure}

\begin{theorem}
If $J(K)$ is Lipschitz continuous, then Algorithm \ref{GDecentAlg} converges to $K^*$, where $K^*$ is a local minimum of $J$.
\end{theorem}

\begin{proof}
This proof is similar to the convergence proof for the gradient descent algorithm for a convex function \cite{boyd2003subgradient}. Since $J(K)$ satisfies the Lipschitz condition, the norm of the gradient vectors, $\vec{g}_i$, are bounded, i.e., there is a $G$ such that $\| \vec{g}_i\|_2 \leq G$ for all $i$, and
\begin{equation}
	|J(K_{i_1})-J(K_{i_2})| \leq G \| K_{i_1}-K_{i_2} \|_2,
\end{equation}
for all $i_1$ and $i_2$. Since $K^*$ is a local minimum of $J$, so there is a neighborhood $U$ containing $K^*$ such that the function is convex on this neighborhood. It is assumed here that $K^j \in U, j=0, ..., i+1$, so we have
\begin{equation}
	\begin{aligned}
	\| K^{i+1}-K^* \|_2^2 = \| K^i - \mu_i \vec{g}_i-K^* \|_2^2 &= \| K^i-K^* \|_2^2 -2 \mu_i \vec{g}_i (K^i-K^*) +\mu_i^2 \| \vec{g}_i \|_2^2  \\
	&\leq \| K^i-K^* \|_2^2 -2 \mu_i \left( J(K^i)-J^* \right) +\mu_i^2 \| \vec{g}_i \|_2^2,
	\end{aligned}
\end{equation}
where, $J^* = J(K^*)$. The inequality above comes from the property of the gradient vector of a convex function which is
\begin{equation}
	J(K^*) \geq J(K^i)+\vec{g}_i (K^*-K^i).
\end{equation}
Applying the inequality above recursively, we have
\begin{equation}
%	\begin{aligned}
	\| K^{i+1}-K^* \|_2^2 \leq \| K^0-K^* \|_2^2 -2 \sum_{j=0}^i \mu_j \left( J(K^j)-J^* \right) + \sum_{j=0}^i \mu_j^2 \| \vec{g}_j \|_2^2,
%	\end{aligned}
\end{equation}
then
\begin{equation}
	2 \sum_{j=0}^i \mu_j \left( J(K^j)-J^* \right) \leq \| K^0-K^* \|_2^2 + \sum_{j=0}^i \mu_j^2 \| \vec{g}_j \|_2^2.
\end{equation}
If we define $J_{sol}^{(i)} = \underset{j=0, \hdots, i}{\min} J(K^j)$ as the solution of the algorithm after $i$ iterations, then we have
\begin{equation}
	J_{sol}^{(i)}-J^* \leq \frac{\| K^0-K^* \|_2^2 + \sum_{j=0}^i \mu_j^2 \| \vec{g}_j \|_2^2}{2 \sum_{j=0}^i \mu_j}.
\end{equation}
Finally, using the assumption $\| \vec{g}_i\|_2 \leq G$, we obtain the inequality
\begin{equation}
	J_{sol}^{(i)}-J^* \leq \frac{\| K^0-K^* \|_2^2 +G^2 \sum_{j=0}^i \mu_j^2}{2 \sum_{j=0}^i \mu_j}.
\end{equation}
If $K^i \in U$ as $i\to\infty$, then by applying the conditions given in \eqref{SSNS}
\begin{equation}
	\lim_{i\to\infty} J_{sol}^{(i)}-J^* \leq \lim_{i\to\infty} \frac{\| K^0-K^* \|_2^2 +G^2 \sum_{j=0}^i \mu_j^2}{2 \sum_{j=0}^i \mu_j} = 0.
\end{equation}
Therefore, the gradient method proposed above converges to the local minimum $K^*$. Otherwise, if at an iteration, say $k < \infty$, $K^k \notin U$, then there is always a finite number of iterations, say $k'$, such that $K^{k+k'} \in U'$, where $U'$ is also a region over which the function is convex, and it contains a local minimum. We can use the same proof given above for this local minimum. Hence, the algorithm converges to a local minimum. 
\end{proof}

The discussion above ensures that the proposed algorithm is convergent. 
Other methods, such as Newton's method, can be used to improve the rate of convergence. But these methods require the computation of the Hessian matrix $\nabla^2 J$ or its approximation. 
Explicit computation of the Hessian matrix is often an expensive process. 
In this direction, approaches such as the quasi-Newton method can be used to estimate the Hessian matrix. 
Any of these higher order methods can be used to achieve higher rate of convergence 
at the cost of increased computational complexity. 
Improving the rate of convergence of first order methods is not the main focus in this paper;
the reader is referred to~\cite{Nocedal99} for such discussions.
%%%%%%%%%%%%%%%%%%%%%%%%%%%%%%%%%%%%%%%%%%%%%%%%%%%%%%%%%%%%%%%%%%%%%%%%%%%
\section{Closed Loop Stability}
\label{sec:stab}

In this section, the stability of the system in the sense of Lyapunov is investigated. 
The idea of consecutive design of stabilizing control by Lyapunov function is used to prove the stability of the system; this approach is inspired by LQR-Trees presented in \cite{tedrake2010lqr}.

\begin{theorem}
Given the cost function \eqref{Cost}, if the transient term is upper bounded by an integrable function $b(t)$:
\begin{equation}
\begin{aligned}
&	l_2(t,\vec{x}) \leq b(t), \ \ \forall t \geq 0, \vec{x} \in \mathbb{R}^n, \vec{u} \in \mathbb{R}^p \\
&	\mathcal{L}(t) = \int_{t}^{\infty} b(\tau) \mathrm{d}\tau < \infty, \ \ \forall t \geq 0\,,
\end{aligned} \label{assumptionsl2}
\end{equation}
and the terminal cost function is positive semi-definite, continuously differentiable, and decreases faster than the value of the function $l_1(\vec{x},\vec{u})$ along all admissible control inputs, $\vec{u}$, and the corresponding trajectories, $\vec{x}$, then the optimal control $\vec{u} = K^* \vec{x}$ will cause the trajectory of the closed loop system to converge to the origin. 
\end{theorem}

\begin{proof}

The admissible function for $l_2(\cdot)$, introduced in \eqref{ObsIdx}, satisfies the conditions given in \eqref{assumptionsl2}. The condition on the terminal cost, $L(t,\vec{x})$ is given as:
\begin{equation}
	\dot{L}(t,\vec{x}) = L_t+L_{\vec{x}} \vec{f}(\vec{x},\vec{u}) \leq -l_1(\vec{x},\vec{u}). \label{assumptionsL}
\end{equation}
Suppose $\vec{u}^*$ and $\vec{x}^*$ are the optimal control and its corresponding trajectory obtained from optimization problem \eqref{CostTran} for $t \in [t_j \ \ t_{j+1})$. To show the stability of this control, a Lyapunov function is defined as:
\begin{equation}
%\begin{aligned}
	V(t) = \int_{t}^{t+\Delta t} \left\{ l_1(\vec{x}^*,\vec{u}^*)-l_2(t,\vec{x}^*) + b(t) \right\} \mathrm{d}t + \mathcal{L}(t+\Delta t) + L(\vec{x}(t+\Delta t),t+\Delta t)\,. \label{lyapunov}
%\end{aligned}	
\end{equation}
Note that all terms in \eqref{lyapunov} are positive, therefore, $V(t) > 0$. Now we need to show that the Lyapunov function \eqref{lyapunov} is decreasing. In fact, the inequality $V(t+\delta) < V(t)$ should hold for any $\delta>0$, such that $t+\delta \in [t_j \ \ t_{j+1}-\Delta t]$. At time $t+\delta$:
\begin{equation}
\begin{aligned}
	V(t+\delta) =& \int_{t+\delta}^{t+\delta+\Delta t} \left\{ l_1(\vec{x}^*,\vec{u}^*) - l_2(t,\vec{x}^*) + b(t) \right\} \mathrm{d}t+ \mathcal{L}(t+\delta+\Delta t) \\
	+& L(\vec{x}(t+\delta+\Delta t),t+\delta+\Delta t)\\
	=& V(t) -\int_{t}^{t+\delta} \left\{ l_1(\vec{x}^*,\vec{u}^*) - l_2(t,\vec{x}^*) + b(t) \right\} \mathrm{d}t + \int_{t+\Delta t}^{t+\delta+\Delta t} \left\{ l_1(\vec{x}^*,\vec{u}^*) \right\} \mathrm{d}t \\
	+& \int_{t+\Delta t}^{t+\delta+\Delta t} \left\{ -l_2(t,\vec{x}^*) + b(t) \right\} \mathrm{d}t + \mathcal{L}(t+\delta+\Delta t) -\mathcal{L}(t+\Delta t) \\
	+& L(\vec{x}(t+\delta+\Delta t),t+\delta+\Delta t) - L(\vec{x}(t+\Delta t),t+\Delta t).
\end{aligned}	
\end{equation}
Based on the assumption \eqref{assumptionsl2}, we have
\begin{equation}
%\begin{aligned}
	\int_{t+\Delta t}^{t+\delta+\Delta t} \left\{ -l_2(t,\vec{x}^*) + b(t) \right\} \mathrm{d}t \leq \int_{t+\Delta t}^{t+\delta+\Delta t} b(t) \mathrm{d}t =\mathcal{L}(t+\Delta t)-\mathcal{L}(t+\delta+\Delta t),
%\end{aligned}	
\end{equation}
and from assumption \eqref{assumptionsL}, one can conclude
\begin{equation}
%\begin{aligned}
	-\int_{t+\Delta t}^{t+\delta+\Delta t} l_1(\vec{x}^*,\vec{u}^*) \mathrm{d}t \geq \int_{t+\Delta t}^{t+\delta+\Delta t} \dot{L} \mathrm{d}t = L(\vec{x}(t+\delta+\Delta t),t+\delta+\Delta t) - L(\vec{x}(t+\Delta t),t+\Delta t).
%\end{aligned}	
\end{equation}
Therefore,
\begin{equation}
%\begin{aligned}
	V(t+\delta) \leq V(t) -\int_{t}^{t+\delta} \left\{ l_1(\vec{x}^*,\vec{u}^*) - l_2(t,\vec{x}^*) + b(t) \right\} \mathrm{d}t < V(t).
%\end{aligned}	
\end{equation}
This proof guarantees the convergence of the closed loop state trajectory $\vec{x}(t)$ to the origin. The stability in the sense of Lyapunov can be proved for all intervals $[t_j \ \ t_{j+1})$. The basin of attraction obtained from Lyapunov function in $[t_j \ \ t_{j+1})$ contains the goal state of the previous interval, $\vec{x}(t_j)$. Therefore, the optimal control obtained here will stabilize the closed loop system.
\end{proof}

Note that the effect of the terminal cost on the stability of the closed loop system has been studied for linear systems as well as nonlinear systems \cite{Fontes01}. By replacing $l_1$ with $l_1+l_2$ in \eqref{assumptionsL}, we obtain the well known condition on the final cost for convergence to the origin of the finite-horizon optimal control problem; this step of our analysis has been inspired by the work of Alessandretti et al. \cite{alessandretti2013model}.

%%%%%%%%%%%%%%%%%%%%%%%%%%%%%%%%%%%%%%%%%%%%%%%%%%%%%%%%%%%%%%%%%%%%%%%%%%%%%%%%
\section{An Illustrative Example}
\label{sec:simulation}

The control synthesis procedure presented in this paper is now illustrated on a holonomic system with a nonlinear measurement. The dynamical system is given by
\begin{equation}
\begin{aligned}
   \dot{x}_1& = u_1\,, \\
   \dot{x}_2& = u_2\,.  \label{holonomic}
\end{aligned}
\end{equation}
In this case, we have a linear system $\dot{\vec{x}} = A\vec{x}+B\vec{u}$, where, $A=0$ and $B=I$. 
By choosing $Q=R=I$, and solving the algebraic Riccati equation $\displaystyle A^T P+P A-P B R^{-1} B^T P+Q=0$, we obtain $P=I$. Thus, the control policy $\displaystyle \vec{u}_{\text{LQR}} = -R^{-1}B^TP\vec{x}=-\vec{x}$ asymptotically stabilizes \eqref{holonomic}. Utilizing this state feedback controller requires full state estimation, which results in having an observable system. 
If any portion of the trajectory is not observable, we have no means of building the feedback controller. 
Due to this coupling between sensing and control in nonlinear systems, we can investigate the observability of the system after applying the controller. 
Assume that the position of the vehicle is continuously measured by an omni-directional camera centered at the origin. Then, the output function is given as 
\begin{equation} \label{NL_obs}
	y = \frac{x_2}{x_1},
\end{equation}
where $y \in \mathbb{R}$ is a bearing only measurement, providing information on 
the direction of the vehicle but not on the distance. 
The observability matrix is given by
\begin{equation}
	d\mathcal{O} =\frac{\partial}{\partial \vec{x}} \begin{bmatrix} y \\ \dot{y} \\ \vdots \end{bmatrix} = \begin{bmatrix} -\frac{x_2}{x_1^2} & \frac{1}{x_1} \\ 2\frac{x_2}{x_1^3}u_1-\frac{1}{x_1^2}u_2 & -\frac{1}{x_1^2}u_1 \\ \vdots & \vdots \end{bmatrix}.
\end{equation}
The system is locally observable if the observability matrix $d\mathcal{O}$, is full rank \cite{nijmeijer1990nonlinear}. The observability matrix for nonlinear system has infinite rows. Here, we can show that there exists a set of control inputs such that the observability matrix becomes full rank. Let us consider the first two rows of the observability matrix; then we have
$$ \det(d\mathcal{O}) = \frac{1}{x_1^3} \left( u_2 -\frac{x_2}{x_1}u_1 \right) \,.$$
Therefore, for all controls $u_2 \neq \frac{x_2}{x_1}u_1$ the observability matrix is full rank, and as a result, the system is observable. Now, if we substitute the optimal control obtained from LQR, $\displaystyle \vec{u}_{\text{LQR}} =-\vec{x}$, we have 
\begin{equation}
	\dot{y} = \frac{\dot{x}_2 x_1- \dot{x}_1 x_2}{x_1^2}= \frac{-x_1x_2+x_1x_2}{x_1^2} = 0.
\end{equation}
All other differential terms are zero ($\ddot{y} = y^{(3)} = \hdots = 0$).
Hence by applying the control $\vec{u}_{\text{LQR}} = -\vec{x}$, the observability matrix is not full rank, and the system is not observable with this choice of control. Therefore, while we have optimal controls for this type of linear system, the nonlinear measurement needs to be managed appropriately.

Now add the observability index to the optimization problem. Assume that the instantaneous cost is given by:
\begin{equation}
l(t)=\vec{x}(t)^T \vec{x}(t) + \vec{u}(t)^T \vec{u}(t) - e^{-t} \text{sat}_{\zeta} \left( \frac{1}{4\epsilon^2} \sum\limits_{i=1}^2  \left( y^{+i}(t)-y^{-i}(t) \right)^2 \right) \,,
\end{equation} 
and assume $t_f=100$. The final cost needs to be chosen to meet the condition \eqref{assumptionsL}. The final cost can be set to $L(\vec{x}_f) = \vec{x}_f^T Q_f \vec{x}_f$, where $Q_f = 0.1 I$. The time span $[t_0, t_f]$ is divided into segments $[0, 1), [1, 2), \cdots, [t_i, t_{i+1}), \cdots, [99, 100)$. Thus, we have 100 intervals, and for each interval the optimal $K_j^*$ for $t \in [t_j\ \ t_{j+1})$ is given by:
\begin{equation} \label{K*ex1}
\begin{aligned}
& K_j^* = \underset{K_j}{\text{argmin}}
& & \int_{t_j}^{t_{j+1}} \left\{ l_1(\vec{x},K_j)-l_2(t,\vec{x},\vec{x}^{\pm 1}, \vec{x}^{\pm 2}) \right\} \mathrm{d}t + \vec{x}(t_{j+1})^T Q_f \vec{x}(t_{j+1}) \,.
\end{aligned}
\end{equation}
Here, we compare the optimal control obtained from the LQR cost function without the observability term and the cost function given in \eqref{K*ex1} which considers the observability of the system. The resulting trajectories obtained from these two scenarios for the initial condition $\begin{bmatrix} -1 \\ 2 \end{bmatrix}$ are
depicted in \cref{CmprTrajs}. The optimal controls for these two scenarios are given in \cref{CTRLs}. As it can be seen in \cref{CTRLs}, initially the absolute values of the observability-based controls are smaller than their corresponding LQR optimal controls. This comes from the fact that by adding the observability term to the optimization, we in fact decrease the cost of being away from the origin. Therefore, the controller chooses a trajectory that is not necessarily the closest path to the origin. However, over time, the observability term becomes less dominant, and the absolute values of the observability-based control take the lead with respect to their corresponding LQR optimal control to satisfy the asymptotic stability condition. The observability-based optimal control initially keeps the system away from an unobservable trajectory, while keeping the shortest path to the origin, (as shown in \cref{CmprTrajs}), and eventually returns to the desired equilibrium. Although, the observability-based controller chooses the longer path, it guarantees the system observability in order to estimate the state information required for the state feedback control.

\begin{figure}[!h]
   \begin{center}
  	 {\includegraphics[width=.75\textwidth]{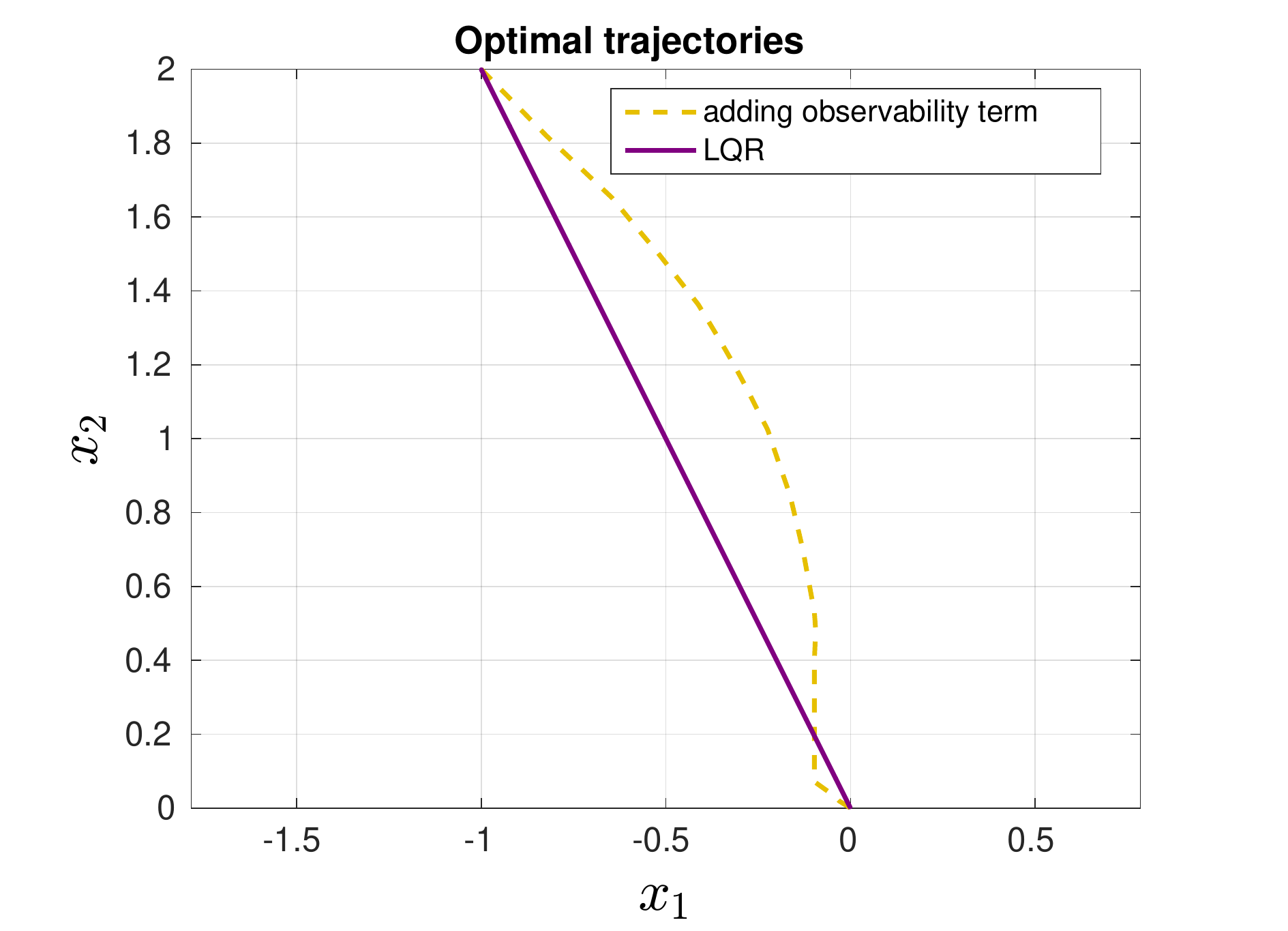}}
	 \vspace{-4 mm}
   \end{center}
   \caption{Trajectories for the system \eqref{holonomic} and measurement \eqref{NL_obs} using LQR optimal state feedback control without (solid, dark color) and with (dashed, light color) observability-based optimal control.}
   \label{CmprTrajs}
\end{figure}

\begin{figure}[!h]
   \begin{center}
  	 {\includegraphics[width=.7\textwidth]{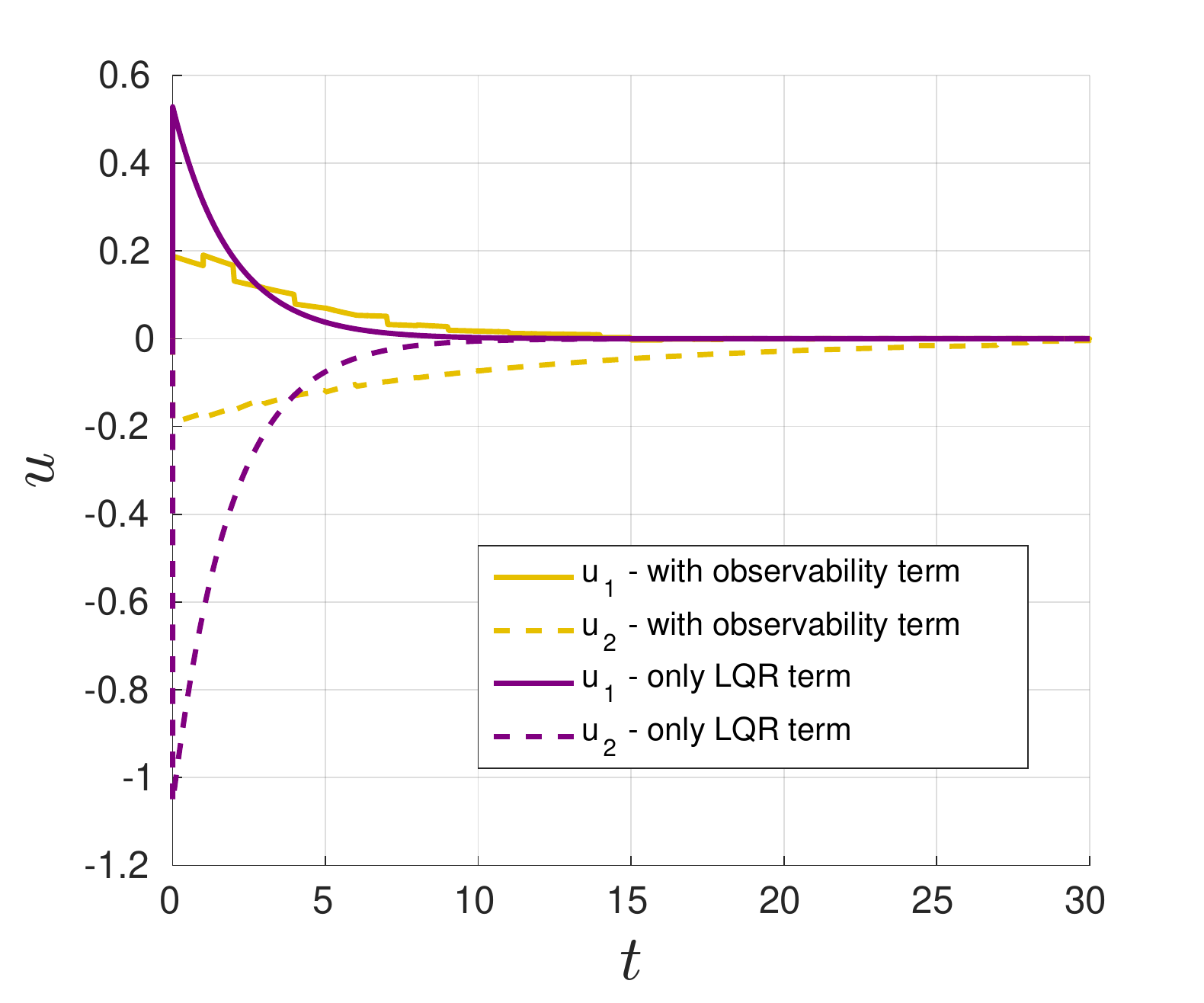}}
	 \vspace{-4 mm}
   \end{center}
   \caption{Optimal controls for the system \eqref{holonomic} and measurement \eqref{NL_obs} using LQR optimal state feedback control without (dark color) and with (light color) observability-based optimal control.}
   \label{CTRLs}
\end{figure}

%%%%%%%%%%%%%%%%%%%%%%%%%%%%%%%%%%%%%%%%%%%%%%%%%%%%%%%%%%%%%%%%%%%%%%%%%%%%%%%%
\section{Conclusion}
\label{sec:conclude}

This paper has been concerned with obtaining a suboptimal control for a nonlinear system based on a nonlinear observability criteria. We proposed an algorithm to optimize a finite-horizon cost function that contains a term that determines asymptotic behavior of the system in addition to a transient term responsible for maximizing a notion of nonlinear observability.
In this direction, the empirical observability Gramian has been used for improving the local observability for nonlinear systems. In addition, stability of the system has been investigated; in particular, it was shown 
that under some conditions on the terminal cost, the stability of the closed loop system is guaranteed.

The cost function was a combination of quadratic and non-quadratic terms and 
motivated by the desire to maximize the observability of the nonlinear system. Similar to other global optimization problems, the main challenge in the proposed approach is the difficulty in computing the global
minimum. 
In order to work around this issue, a gradient-based recursive algorithm was used to obtain a piecewise optimal linear state feedback controller. 
Linear state feedback is not necessarily optimal control for general nonlinear systems. Furthermore, not all classes of nonlinear systems (e.g., nonholonomic systems) can be stabilized by smooth feedback control \cite{Lin95, Phat06} or by time invariant state feedback control. However, such systems can generally be stabilized by time varying oscillatory control, and the next step in our work is considering appropriate time varying oscillatory feedback control for stabilization while improving the observability of the system.

\section*{ACKNOWLEDGMENTS}
The authors thank Bill and Melinda Gates for their active support and their sponsorship through the Global Good Fund.

\section*{References}

\bibliography{citations}

\end{document}